%% file: main.tex
\documentclass[acmsmall]{acmart}
\pdfoutput=1

\input{includes}

 \acmConference[]{}{}{}

\setcopyright{rightsretained}

\begin{document}

\title{Version Space Algebras are Acyclic Tree Automata}


\author{James Koppel}
\affiliation{
  \institution{MIT}            
  \city{Cambridge}
  \state{MA}
  \country{USA}
}
\email{jkoppel@mit.edu}          

\begin{abstract}
Version space algebras are ways of representing spaces of programs which can be combined using union, intersection, and cross-product/``join" operators. In their reified form as ASTs with explicit union and join nodes, they have the ability to compactly represent exponentially-large spaces of programs, owing to which they have become become the most popular approach to enumerative program synthesis since the introduction of FlashFill in 2010.  We present a linear-time semantics-preserving constructive embedding from version space algebras into nondeterministic finite tree automata, showing that the former are but a special case of the latter. Combined with recent results finding a correspondence between e-graphs and minimal deterministic tree automata, this shows that tree automata are strict generalizations of all recent major approaches to efficiently representing large spaces of programs by sharing.
\end{abstract}




\maketitle

\section{Introduction}

Enumerative program synthesis is the technique of listing a large number of programs and choosing the best one; deductive program synthesis is the process of starting with description and finding the right set of rewrites to transform it into the desired program. Since there can be exponentially many terms and rewrite sequences of a given size, both are greatly aided by compact representations of large spaces of terms.

In the past decade, two approaches have gained great popularity for solving this problem, version space algebras (VSAs) and equivalence graphs (e-graphs), with the latter also being applied to semantic code search \cite{premtoon2020semantic} and equivalence-checking \cite{alet2021large,stepp2011equality}. More recently, Wang \cite{wang2017synthesis} applied a third and very old data structure to the same problem of representing spaces of terms: tree automata.

We give an example of each. Consider representing the set of $9$ terms $\mathcal{T}=\{f(g(X),g(Y))\}$, where $X,Y \in \{a,b,c\}$.  \figref{fig:ex1} gives the e-graph, VSA, and tree automaton representing this set.

But the similarity between the three is striking, and they all appear to function by sharing subterms in the same way. Though works involving VSAs, e-graphs, and tree automata typically construct and use them very differently, their representational capabilities appear quite similar.

In fact, we shall see that their capabilities are not only similar but identical, with tree automata being strict generalizations of the other two, which correspond to restricted forms of tree automata. Pollock and Haan \cite{pollock_e-graphs_nodate} provided half this story in a recent write-up, showing that e-graphs are isomorphic to minimal determinstic finite tree automata (minimal DFTAs), with the well-known congruence-closure algorithm used in e-graph maintenance being equivalent to tree-automata minimization. In this short paper, we show the other half, that version space algebra correspond to acyclic nondeterministic tree automata.

\begin{figure}
\centering
\begin{subfigure}[t]{0.34\textwidth}
\centering
\includegraphics[scale=0.45]{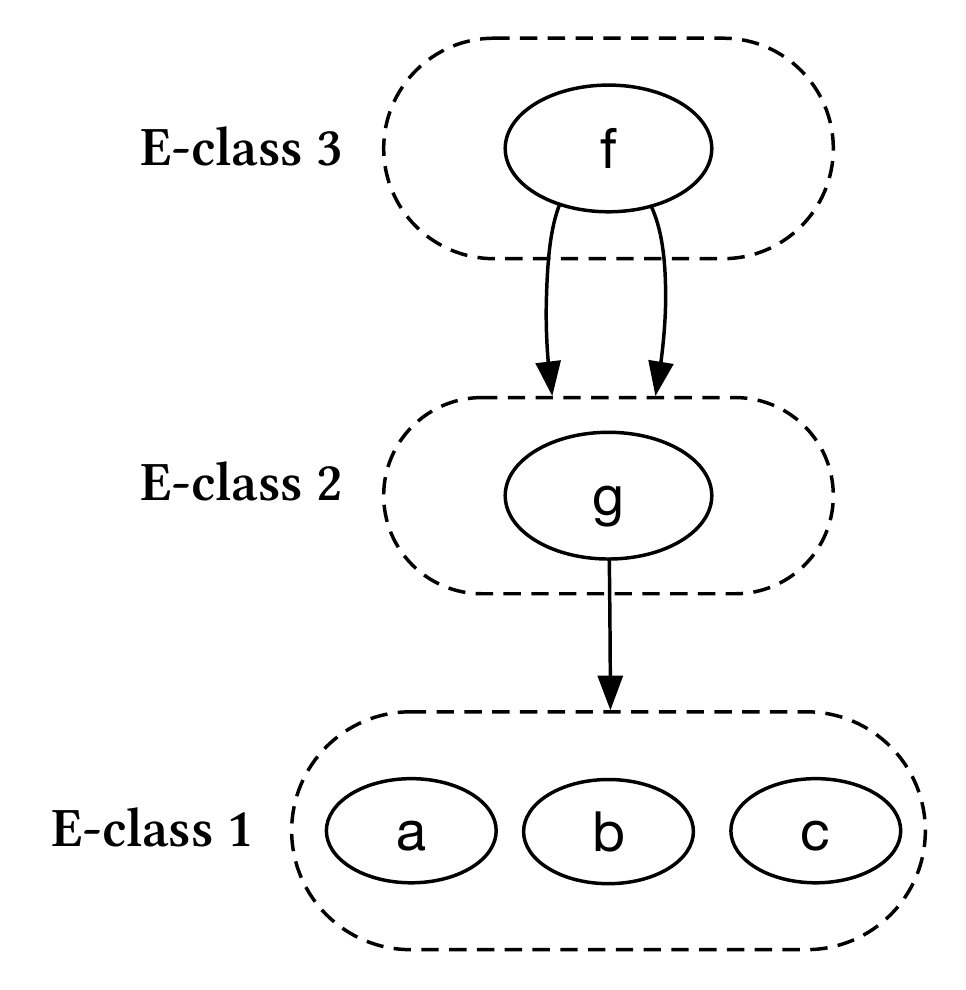}
\subcaption{E-graph}
\label{fig:ex1-egraph}
\end{subfigure}
~ 
\begin{subfigure}[t]{0.31\textwidth}
\centering
\includegraphics[scale=0.45]{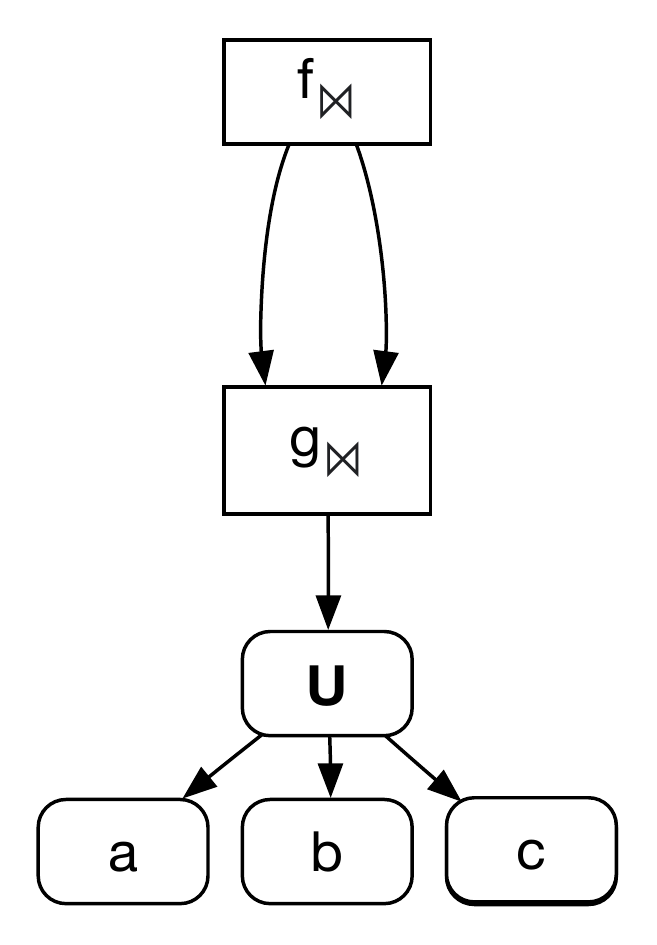}
\subcaption{Version-space algebra}
\label{fig:ex1-vsa}
\end{subfigure}
~
\begin{subfigure}[t]{0.25\textwidth}
\centering
\includegraphics[scale=0.45]{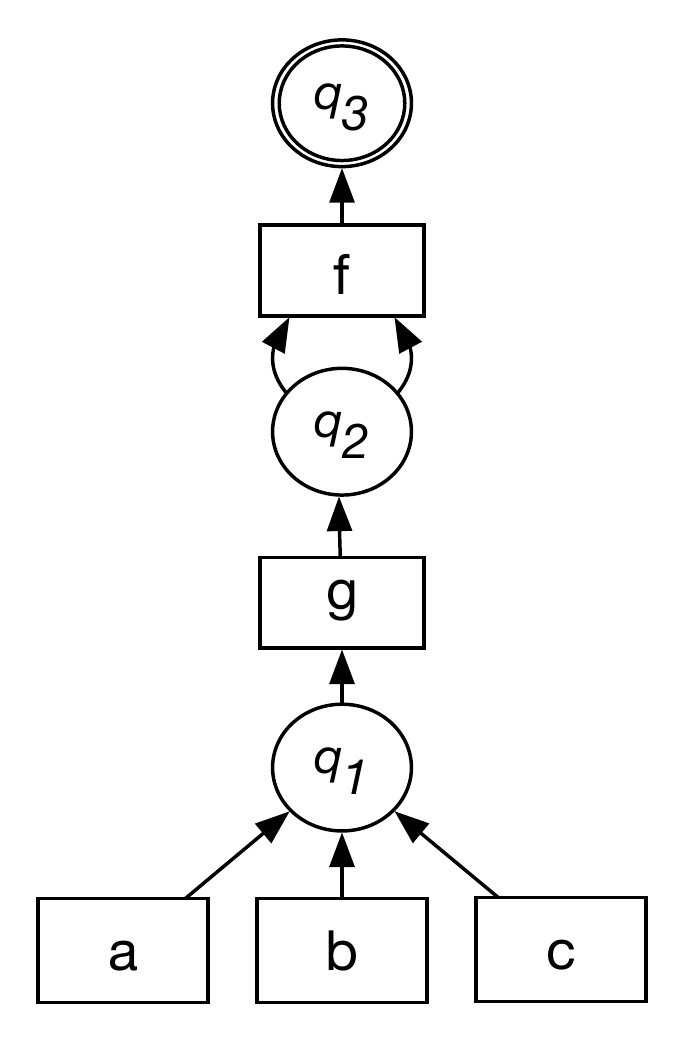}
\subcaption{Tree automaton. Rectangles represent transitions.}
\label{fig:ex1-dfta}
\end{subfigure}
\caption{}
\label{fig:ex1}
\end{figure}

There has been foreshadowing in the literature of a connection between VSAs and automata. The FlashFill paper \cite{gulwani2011automating}, which sparked the wave of synthesis papers based on version-space algebras, featured an intersection algorithm for associative VSAs immediately noted to be similar to the intersection of string automata.  Correspondences between schemes for sharing substructure and automata go at least back to the discovery of the correspondence between binary decision diagrams and string automata  \cite{michon1998automata}. Yet, though the similarities between VSAs and tree automata may be obvious to one who digests \figref{fig:ex1}, this paper appears to be the first to state it publicly and give a formal correspondence.

\section{Preliminaries}

We now present formal definitions of both version space algebras and tree automata. In both developments, we shall use $\Sigma$ to denote a signature, a set of symbols. Both VSAs and tree automata represent sets of terms over a given signature $\Sigma$. We do not formally restrict symbols to be of a given arity. We shall also provide a normalized form of version space algebras which makes their correspondence with tree automata particularly simple.

\subsection{Version Space Algebra}

Abstractly speaking, a version space algebra (VSA) is a space of programs equipped with union, intersection, and independent join operators.  \cite{lau2003programming} In the past decade however, they have become synonymous with their typical concrete representation as ASTs augmented with explicit union and join nodes, used for compactly storing large spaces of programs in enumerative synthesis. In this form, they can be considered the general version of searching program spaces by dynamic programming, as they achieve their efficiency by sharing common subterms of programs within the space. The FlashMeta/PROSE paper \cite{polozov2015flashmeta} is the definitive work for VSAs as currently used; we replicate their formalism in brief here.

\begin{definition}[Version space algebra (Orthodox definition)]
    \label{def:vsa}
    Let $N$ be a symbol in a signature $\Sigma$.
    A \emph{version space algebra} over $\Sigma$ is a representation for a set $\vsa$ of terms rooted at $N$. The grammar of VSAs is:
    \vspace{1pt}
    \begin{align*}
        \vsa \is \left\{P_1, \dots, P_k\right\} \palt \vsaunion(\vsa_1, \dots, \vsa_k) \palt \vsajoin(\vsa_1, \dots, \vsa_k)
    \end{align*}
    where $F$ is any $k$-ary node in $\Sigma$, and $P_j$ are terms generated by signature $\Sigma$.
    The semantics of a VSA as a set of programs is given as follows:
    \vspace{0.5em}
    \begin{eqnarray*}
        \denotation{\left\{P_1, \dots, P_k\right\}}^\VSA &=& \left\{P_1, \dots, P_k\right\} \\
        \denotation{\vsaunion(\vsa_1, \dots, \vsa_k)}^\VSA &=& \bigcup_{i=1}^k \denotation{\vsa_i}^\VSA \\
        \denotation{\vsajoin(\vsa_1, \dots, \vsa_k)}^\VSA &=& \left\{ F(P_1, \dots, P_k) \bigmid P_i \in \denotation{\vsa_i}^\VSA \right\}
    \end{eqnarray*}
    \vspace{0.5em}
\end{definition}

While syntactically represented as trees, VSAs are treated as DAGs by automatically sharing all identical subterms. The requirement that all duplicate nodes are shared, while semantically unnecessary, becomes quite important when analyzing the embedding into tree automata. We hence modify our presentations to annotate every distinct $\vsaunion$ and $\vsajoin$ node with a unique explicit label $l$. The VSA must satisfy the \textbf{uniqueness property} that no two distinct nodes may have the same children. Further, we shall henceforth assume that these are the only two node types in a VSA:

\begin{observation}
If a language has a finite set of constructors and constants, then the set nodes of VSAs are redundant. All program sets can be expressed using union and join nodes, e.g.: $\left\{f(A, g(B)), g(C)\right\}$ can be represented $\vsaunion(f(A(), g(B())), g(C()))$, treating $A, B$, and $C$ as nullary constructors.
\end{observation}

Finally, note that any join node $\vsajoin(\vsa_1, \dots, \vsa_k)$ can be replaced with a unary union node, namely  $\vsaunion(\vsajoin(\vsa_1, \dots, \vsa_k))$. Combining these observations, we can thus transform any VSA to be composed of strictly-alternating layers of union and join nodes. We use the symbol $\P{\vsa}$ (i.e.: sets of VSAs) to represent the union-node layers. Here is a definition for these normalized VSAs; the semantics are the same as in orthodox VSAs.

\begin{definition}[Version space algebra (Normalized definition)]
    \label{def:vsa}
    A \emph{normalized version space algebra} is a representation for a set $\vsa$ of terms. The grammar of normalized VSAs is:
    \vspace{1pt}
    \begin{align*}
        \vsa &\is \vsajoin^l(\P{\vsa}_1, \dots, \P{\vsa}_k) \\
        \P{\vsa} &\is \vsaunion^l(\vsa_1, \dots, \vsa_k)
    \end{align*}
\end{definition}

\subsection{Tree Automata}

Tree automata are generalizations of ordinary string automata to trees. Tree automata consume a tree bottom-to-top, associating a state with every node, choosing the state for each node based on the states of its children. When every node is unary except for a special ``start" symbol, they behave identically to string automata. Tree automata can both be considered operationally, as a machine that inputs a tree over a signature $\Sigma$ and outputs ``accept" or ``reject," and denotationally, as a representation of a set (or language) of trees. The definitive treatment of tree automata is Comon \cite{comon1997tree}.

We present our definition below, along with that of the denotation of a tree automaton:

\begin{definition}[(Bottom-up) Nondeterministic Finite Tree Automata (NFTA)]
A (bottom-up) finite tree automaton over signature $\Sigma$ is a tuple $\mathcal{A}=(\mathcal{Q}, \Sigma, \mathcal{Q}_f, \Delta)$ where $\mathcal{Q}$ is a set of states, $\mathcal{Q}_f \subset \mathcal{Q}$ is a set of final states, and $\Delta \subseteq \Sigma \times \mathcal{Q}^* \times \mathcal{Q}$ is a set of transitions of the form $f(q_1, \dots, q_k) \rightarrow q$.

The denotation of a transition and state, both within a larger automaton with transitions $\Delta$, are given mutually-recursively by:

\begin{eqnarray*}
\denotation{q}^\TA_\Delta &=& \bigcup \left\{ \denotation{F(q_1, \dots, q_k) \rightarrow q}^\TA_\Delta \bigmid (F(q_1, \dots, q_k) \rightarrow q) \in \Delta\right\} \\
\denotation{F(q_1, \dots, q_k)\rightarrow q}^\TA_\Delta &=& \left\{ F(P_1, \dots, P_k) \bigmid P_i \in \denotation{q_i}^\TA_\Delta \right\}
\end{eqnarray*}

\end{definition}

\section{Embedding}

In this section, we present the embedding of a version space algebras into tree automaton, and prove that the resulting automaton represents the same set of terms.

\begin{definition}[Embedding of VSA into TA]
\label{defn:embedding}
Let $\vsa$ be a normalized VSA over $\Sigma$. We shall construct a tree automaton corresponding to $\vsaunion^r(\vsa)$, where $r$ is a new ``root" label.

Let $\subs{\vsaunion^r(\vsa)}$ be the set of subnodes of $\vsaunion^r(\vsa)$. For each union node $\vsaunion^l(\dots) \in \subs{\vsaunion^r(\vsa)}$, define the automaton state $q^l$. Let the set of all such states be $\mathcal{Q}$. We define the transitions as

\[
  \Delta = \left\{ F(q^{l_1}, \dots, q^{l_k}) \rightarrow q^l  \bigmid \vsaunion^l(\dots, \vsajoin[F](\vsaunion^{l_1}(\dots), \dots, \vsaunion^{l_k}(\dots)), \dots) \in \subs{\vsajoin^r(\vsa)} \right\}
\]

Finally, we define

\begin{align*}
\vsatota{\vsa} = (\mathcal{Q}, \Sigma, \left\{q^r\right\}, \Delta)
\end{align*}

Note that the resulting automaton is not necessarily deterministic, as the same join node may be a child of distinct union nodes. It is, however, acyclic, as VSAs are acyclic.

\end{definition}

\begin{remark}
The construction of Definition \ref{defn:embedding} creates one transition in the tree automaton per edge between a join and union in the VSA, and each such transition may have as children a significant fraction of states in the automaton. Naively implemented, for a VSA with $V$ nodes and $E$ edges, it is thus actually $O(VE)$ in time and space, although the number of states and transitions is still $O(V+E)$. However, this can be accelerated by giving identity to repeated transitions. That is, given the set of transitions

\[
  F(q_1, \dots, q_k) \rightarrow r_1, \dots, F(q_1, \dots, q_k) \rightarrow r_m
\]

One can represent these by the factored representation $\delta_i = F(q_1, \dots, q_k)$ and $\delta_i \rightarrow r_1, \dots, \delta_i \rightarrow r_m$, where $i$ is some new label, thus restoring the implementation to $O(V+E)$ time and space. Graphically, this appears as multiple states of the tree automaton sharing the same inbound transitions (hyperedges).
\end{remark}

\begin{theorem}[Correctness of Embedding]
\[ 
  \denotation{\vsa}^\VSA = \denotation{\vsatota{\vsa}}^\TA
\]
\end{theorem}
\begin{proof}
Let $\mathcal{Q}$ be the states of $\vsatota{\vsa}$, with union nodes $\vsaunion^l$ corresponding to states $q^l$. We show by mutual induction that

\[
  \denotation{\vsaunion^j(\dots)}^\VSA=\denotation{q^j}^\TA_\Delta
\]

and that, for any $q$,

\[
  \denotation{\vsajoin(\vsaunion^{i_1}(\dots), \dots, \vsajoin(\vsaunion^{i_k}(\dots))}^\VSA=\denotation{q^i}^\TA_\Delta
\]

The theorem immediately follows.

First, for join nodes:

\begin{align*}
    \denotation{\vsajoin(\vsaunion^{i_1}(\dots), \dots, \vsaunion^{i_k}(\dots)}^\VSA &=  \left\{ F(P_1, \dots, P_k) \bigmid P_j \in \denotation{\vsaunion^{i_j}}^\VSA \right\} \\
    &=  \left\{ F(P_1, \dots, P_k) \bigmid P_j \in \denotation{q^{i_j}}^\TA_\Delta \right\} \\
    &= \denotation{F(q^{i_1}, \dots, q^{i_k}) \rightarrow q}^\TA_\Delta
\end{align*}

Now, for union nodes:

\begin{align*}
    \denotation{\vsaunion^l(\vsajoin^{i_1}(\dots), \dots, \vsajoin^{i_k}(\dots))}^\VSA &=  \bigcup_{j=1}^k \denotation{F^{i_j}(\dots)}^\VSA  \\
    &= \bigcup_{j=1}^k \denotation{F^{i_j}(\dots) \rightarrow q^l}^\TA_\Delta \\
    &= \denotation{q^l}^\TA_\Delta
\end{align*}

\end{proof}

\section{Conclusion}

E-graphs became very popular in 2020-2021 with the publication of \textsc{egg} \cite{willsey2021egg} and a number of deductive synthesis/transformation papers based on it. By connecting e-graphs to tree automata, Pollock and Haan allowed a wealth of prior work to be applied to e-graphs. In particular, they solved the problem of e-graph intersection, making them usable in inductive synthesis. \cite{pollock_e-graphs_nodate}

But, through another lens, they showed that much of the excitement around e-graphs may be better directed at tree automata. Our work furthers this story, showing that tree automata were secretly behind another of the past decade's great advances in program synthesis, and promising to unite both communities in continuing to gain new techniques by study of tree automata.

Finally, a word to all data structure designers: if ever you encounter a tree-like or DAG-like data structure with alternating levels of conjunction-esque and disjunction-esque nodes, you may be dealing with a tree automaton in disguise.

\bibliography{citations.bib}

\end{document}

%% file: includes.tex
\usepackage{mathpartir}
\usepackage{proof}

\usepackage{amsmath}
\usepackage{enumitem}
\usepackage{thmtools}
\usepackage{thm-restate}
\usepackage{bm}
\usepackage{booktabs}
\usepackage{framed}
\usepackage[utf8]{inputenc}
\usepackage{mathtools}
\usepackage{mdframed}
\usepackage{graphicx}
\usepackage{multirow}
\usepackage{stmaryrd}
\usepackage{subcaption}
\usepackage{tabularx}
\usepackage{wasysym}
\usepackage{xspace}

\usepackage{cmap}
\usepackage[T1]{fontenc}

\usepackage{wrapfig}

\usepackage[textsize=tiny]{todonotes}

\usepackage[scaled=0.75]{beramono}
\usepackage{nanoml}

\newcommand{\figref}[1]{Fig. \ref{#1}}

\newtheorem{remark}{Remark}
\newtheorem{observation}{Observation}

\newcommand{\set}[1]{\ensuremath{\textsf{#1}}}

\newcommand{\denotation}[1]{\llbracket #1 \rrbracket}

\renewcommand{\P}[1]{\mathbb{P}(#1)}

\newcommand{\TA}{\set{TA}}

\newcommand{\VSA}{\textsc{VSA}}
\newcommand{\vsa}[1][N]{\widetilde{#1}}
\newcommand{\vsajoin}[1][F]{#1_{\text{\tiny\Bowtie}}}
\newcommand{\vsaunion}{\pmb{\mathsf{U}}}

\newcommand{\is}{\;:=\;}
\newcommand{\palt}{\;|\;}

\newcommand{\bigmid}{\mathrel{\big|}}

\newcommand{\vsatota}[1]{\mathsf{VSAToTA}(#1)}
\newcommand{\subs}[1]{\mathsf{subs}(#1)}